\documentclass[reqno,11pt]{amsart}

\usepackage[T1]{fontenc}
\usepackage[utf8]{inputenc}
\usepackage[english]{babel}

\usepackage[foot]{amsaddr}

\usepackage{mathtools,amssymb,amsthm}
\usepackage{enumitem}
\usepackage{caption,subcaption}
\newsavebox{\largestimage}

\usepackage{multirow}
\usepackage[x11names]{xcolor}
\usepackage{tikz}
\usetikzlibrary{shapes}
\tikzstyle{every picture}=[line join=round,line cap=round,every label/.append style={font=\small},label distance=-1pt,line width=.8pt]

\setenumerate{label={{\rm({\roman*})}},leftmargin=8mm,itemsep=3pt,topsep=3pt}
\setitemize{label={$\vcenter{\hbox{\tiny$\bullet$}}$},leftmargin=8mm,itemsep=3pt,topsep=3pt}

\usepackage[a4paper,tmargin=2.5cm,rmargin=2cm]{geometry}

\definecolor{darkgreen}{RGB}{50,180,50}

\usepackage{hyperref}

\makeatletter
\renewcommand{\email}[2][]{%
  \ifx\emails\@empty\relax\else{\g@addto@macro\emails{,\space}}\fi%
  \@ifnotempty{#1}{\g@addto@macro\emails{\textrm{(#1)}\space}}%
  \g@addto@macro\emails{#2}%
}
\makeatother

\makeatletter
\newtheorem*{rep@theorem}{\rep@title}
\newcommand{\newreptheorem}[2]{%
\newenvironment{rep#1}[1]{%
 \def\rep@title{#2 \ref{##1}}%
 \begin{rep@theorem}}%
 {\end{rep@theorem}}}
\makeatother

\newreptheorem{theorem}{Theorem}
\newreptheorem{corollary}{Corollary}

\newtheorem{theorem}{Theorem}
\newtheorem{conjecture}[theorem]{Conjecture}
\newtheorem{observation}[theorem]{Observation}

\theoremstyle{remark}

\newcommand{\R}{\mathbb{R}}

\renewcommand{\le}{\leqslant}

\renewcommand{\leq}{\leqslant}
\renewcommand{\geq}{\geqslant}
\renewcommand{\preceq}{\preccurlyeq}

\DeclareMathOperator{\xc}{\mathsf{xc}}
\DeclareMathOperator{\COR}{\mathsf{COR}}
\DeclareMathOperator{\STAB}{\mathsf{STAB}}

\DeclareMathOperator{\tw}{\mathsf{tw}}
\DeclareMathOperator{\LP}{\mathsf{LP}}

\newcommand{\set}[2]{\bigl\{\,#1 \bigm\vert#2\,\bigr\}}

\begin{document}

\title{Extension Complexity of the Correlation Polytope}

\author[P.~Aboulker]{Pierre Aboulker\textsuperscript{1}}
\email[1]{pierreaboulker@gmail.com}
\address[1]{Laboratoire G-SCOP, CNRS, Universit\'e Grenoble Alpes}

\author[S.~Fiorini]{Samuel Fiorini\textsuperscript{2}}
\address[2,3,4]{D\'epartement de Math\'ematique, Universit\'e libre de Bruxelles}
\email[2]{sfiorini@ulb.ac.be}  
\author[T.~Huynh]{Tony Huynh\textsuperscript{3}}
\email[3]{tony.bourbaki@gmail.com}
\author[M.~Macchia]{Marco Macchia\textsuperscript{4}}
\email[4]{mmacchia@ulb.ac.be}
\author[J.~Seif]{Johanna Seif\textsuperscript{5}}
\address[5]{D\'epartement d'Informatique, Ecole Normale Sup\'erieure de Lyon}
\email[5]{johanna.seif@ens-lyon.fr}

\begin{abstract}
We prove that for every $n$-vertex graph $G$, the extension complexity of the correlation polytope of $G$ is $2^{O(\tw(G) + \log n)}$, where $\tw(G)$ is the treewidth of $G$. Our main result is that this bound is tight for graphs contained in minor-closed classes.  
\end{abstract}


\maketitle


\section{Introduction}
 All graphs in this paper are undirected and simple. Let $G=(V,E)$ be a graph.  The \emph{correlation polytope} of $G$, denoted $\COR(G)$, is the convex hull of the characteristic vectors of induced subgraphs of $G$. More precisely, $\COR(G)$ is the polytope in $\mathbb{R}^{V \cup E}$ which is the convex hull of all vectors of the form $(\chi(X), \chi(E(X)))$ for $X \subseteq V$, where $E(X)$ is the set of edges with both endpoints in $X$, and $\chi$ denotes the characteristic vector of a set. It is closely related to the \emph{stable set polytope} of $G$, denoted $\STAB(G)$, which is the convex hull of characteristic vectors of stable sets of $G$.
 There are other equivalent definitions of the correlation polytope, and it arises naturally in many different contexts, including probability theory, propositional logic, and quantum mechanics~\cite{Pitowsky91}.  
 
 More recently, the correlation polytope has also acquired greater prominence in machine learning, where it is more commonly referred to as the \emph{marginal polytope}. For example, Wainright and Jordan~\cite{WJ08} showed that the \emph{maximum a posteriori} (\textsf{MAP}) inference problem for graphical models~\cite{LPC12,DDA14,EDM17, RW17} is equivalent to maximizing a linear function over the correlation polytope.  This motivates the search for \emph{compact} descriptions of the correlation polytope.  The proper framework 
 for addressing such questions is the theory of \emph{extended formulations}~\cite{Kaibel11,CGZ13, CCZ14,Roughgarden15,Weltge15}.

A polytope $Q \subseteq \R^p$ is an \emph{extension} of a polytope $P \subseteq \R^d$ if there exists an affine map $\pi: \R^p \rightarrow \R^d$ with $\pi(Q) = P$. The \emph{extension complexity} of $P$, denoted $\xc(P)$, is the minimum number of facets of any extension of $P$. If $Q$ is an extension of $P$ such that $Q$ has fewer facets than $P$, then it can be advantageous to optimize over $Q$ instead of $P$.

Fiorini, Massar, Pokutta, Tiwary and de Wolf~\cite{fiorini2011exponential} were the first to show that many polytopes associated to classic $\mathsf{NP}$-hard problems, (including the correlation polytope of the complete graph) have exponential extension complexity. Their results do not rely on any complexity assumptions such as $\mathsf P \neq \mathsf{NP}$ or $\mathsf{NP} \not\subseteq \mathsf{P}/\mathsf{poly}$.  In this paper, we determine the extension complexity of the correlation polytope almost exactly. 

In order to state our main results, we need some graph theoretic definitions.  
A graph $H$ is a \emph{minor} of a graph $G$, denoted $H \preceq G$, if $H$ can be obtained from a subgraph of $G$ by contracting edges. A class $\mathcal{C}$ of graphs is \emph{minor-closed} if $G \in \mathcal{C}$ and $H \preceq G$ implies $H \in \mathcal{C}$. In Observation~\ref{obs:minormonotone}, we note that if $H \preceq G$, then $\xc(H) \leq \xc(G)$.  By the graph minor theorem of Robertson and Seymour~\cite{RS04}, the property $\xc(\COR (G)) \leq k$ is characterized by a finite set of forbidden minors.  For more on the connection between the correlation polytope and graph minors, see~\cite{Weller16, WRS16}.

A \emph{tree-decomposition} of a graph $G$ is a pair $(T, \mathcal{B})$
where $T$ is a tree and $\mathcal{B} \coloneqq \set{B_t }{ t \in V(T)}$ is a collection of subsets of vertices of $G$ satisfying:

\begin{itemize}
\item $V(G)= \bigcup_{t \in V(T)} B_t$, 
\item for each $uv \in E(G)$, there exists $t \in V(T)$ such that $u,v \in B_t$, and
\item for each $v \in V(G)$, the set of all $w \in V(T)$ such that $v \in B_w$ induces a subtree of $T$.  
\end{itemize}
The \emph{width} of the decomposition $(T, \mathcal{B})$ is $\max \set{|B_t|-1 }{ t \in V(T)}$.  The \emph{treewidth} of $G$, denoted $\tw(G)$, is the minimum width taken over all tree-decompositions of $G$.     

Wainright and Jordan~\cite{WJ04} proved that for all graphs $G$, the $(\tw(G)+1)$-th level of the \emph{Sherali-Adams hierarchy}~\cite{SA90} of the $\LP$ relaxation of $\COR(G)$ is equal to $\COR(G)$.  It follows that $\COR(G)$ has an extended formulation of size $n^{O(\tw(G))}$.  We now derive an improved upper bound using results of Kolman and Kouteck\'{y}~\cite{KK15} (see also Bienstock and Munoz~\cite{BM15}). 

\begin{theorem}\label{thm:upperbound}
For every $n$-vertex graph $G$, the extension complexity of $\COR(G)$ is $2^{O(\tw(G) + \log n)}$.
\end{theorem} 

\begin{proof}
Let $G$ be a graph with $n$ vertices and $m$ edges. Note that $\COR(G)$ is the convex hull of all $0/1$-vectors $x \in \R^{V \cup E}$ satisfying $x_{uv}=x_ux_v$ for all $uv \in E(G)$. We define the \emph{constraint graph} of $\COR(G)$ to be the graph $G'$ whose vertices are the variables of the above system, where two variables are adjacent in $G'$ if and only if they appear in a common constraint.  By~\cite[Theorem 1]{KK15}, $\COR(G)$ has an extended formulation of size $(n+m)2^{\tw(G')}$. Observe that $G'$ is obtained from $G$ by adding a path of length $2$ between $u$ and $v$ for all $uv \in E(G)$. It is easy to see that if $\tw(G)=1$, then $\tw(G')=2$ and if $\tw(G) \geq 2$, then $\tw(G')=\tw(G)$.  Therefore, $(n+m)2^{\tw(G')}=2^{O(\tw(G) + \log n)}$, as required. 
 \end{proof}

Up to the constant factor in the exponent, we conjecture that our upper bound is tight. 

\begin{conjecture}
For every $n$-vertex graph $G$, the extension complexity of $\COR(G)$ is $2^{\Omega(\tw(G) + \log n)}$.
\end{conjecture}

G\"o\"os, Jain and Watson \cite{Goos18} proved that there exists a sequence $(G_n)$ of graphs where each $G_n$ is an $n$-vertex graph and the extension complexity of $\STAB(G_n)$ is $2^{\Omega(n / \log n)}$. If true, our conjecture would improve their bound to $2^{\Omega(n)}$
and it would yield explicit $0/1$-polytopes with extension complexity exponential in their dimension, as predicted by the counting argument of Rothvo\ss~\cite{Rothvoss13} (see~\cite[Problem 7]{Weltge15}). 

As evidence for the conjecture, we prove that our bound is tight for minor-closed classes of graphs.  The following is our main result.  

\begin{theorem} \label{lowerbound}
For every proper minor-closed class $\mathcal{C}$, there exists a constant $c>0$ such that for every $n$-vertex graph $G \in \mathcal{C}$,
\begin{equation*}
\xc(\COR(G)) \geq 2^{c (\tw(G)+\log n)}\,.
\end{equation*}
\end{theorem}

Actually, the proof of Theorem~\ref{lowerbound} shows that for \emph{every graph} $G$, we have $\xc(\COR(G)) \geq 2^{\Omega(h + \log n)}$, where $h$ is the maximum height of a grid that $G$ contains as a minor. In virtue of the polynomial grid-minor theorem of Chekuri and Chuzhoy~\cite{chekuri2016polynomial, chuzhoy2016}, this implies that $\xc(\COR(G)) \geq 2^{\Omega(\tw(G)^\delta + \log n)}$ for some universal constant $\delta > \frac{1}{20}$.

\section{The lower bound}
In this section we prove our lower bound.  Let $P$ be a polytope. We note the following easy (and folklore) facts. If $A$ is an affine subspace, then $\xc(P \cap A) \leq \xc(P)$, and if $\pi$ is an affine map, then $\xc(\pi(P)) \leq \xc(P)$.  Since the extension complexity of a polytope is at least its dimension, we also have the following easy observation.

\begin{observation} \label{dimension}
For all $n$-vertex graphs $G$,
%
\(\xc(\COR(G)) \geqslant n\).
%
\end{observation}

We next show that the extension complexity of the correlation polytope is monotone under taking minors.  

\begin{observation} \label{obs:minormonotone}
If $H \preceq G$, then $\xc(\COR(H)) \leqslant \xc(\COR(G))$.
\end{observation}

\begin{proof}
It is easy to see that $H \preceq G$ if and only if $H$ can be obtained from $G$ by deleting edges, contracting edges, and removing isolated vertices. We show that none of these operations increases the extension complexity of the correlation polytope. 

Let $uv \in E(G)$. Then $\COR(G \setminus uv)$ can be obtained from $\COR(G)$ by projecting out $x_{uv}$. Moreover $\COR(G / uv)$ is obtained from $\COR(G)$ by setting $x_u=x_{uv}=x_v$ (this defines a face since $x_{uv} \leq x_u$ and $x_{uv} \leq x_v$ are valid). If $w$ is an isolated vertex of $G$, then $\COR(G-w)$ is obtained from $\COR(G)$ by projecting out $x_w$.
\end{proof}

For $h \in \mathbb{N}$, we let $K_{h, h}$ be the complete bipartite graph with $2h$ vertices and $h^2$ edges, and $G_{h, h}$ be the $h \times h$ grid.  Recall that $G_{h,h}$ has vertex set $[h] \times [h]$, where $(a,b)$ is adjacent to $(a',b')$ if and only if $|a-a'|+|b-b'|=1$. 

We now define a modified grid that will appear in the proof of Theorem~\ref{thm:lowerbound_new}.
A \emph{grid with gadgets of height $h-1$} is obtained by modifying the grid $G_{h,h}$ as follows.  Let $G_{h,h}^{\circ}$ be the graph obtained from $G_{h,h}$ by subdividing each edge of $G_{h,h}$. For each $i, j \in [h]$, let $r_{i,1}, r_{i,1}', r_{i,2}, r_{i,2}', \dots, r_{i,h-1}',r_{i, h}$ be the vertices of $G_{h,h}^{\circ}$ along the $i$th row and $r_{1,j}, c_{1,j}', r_{2,j}, c_{2,j}', \dots, c_{h-1,j}',r_{h,j}$ be the vertices of $G_{h,h}^{\circ}$ along the $j$th column. 
The grid with gadgets of height $h-1$ is obtained from $G_{h,h}^{\circ}$ by performing the following operations:
\begin{itemize}
    \item delete $r_{1,1}$ and $r_{h,h}$,
    \item delete $r_{1,1}', r_{1,2}', \dots, r_{1,h-1}'$ and $c_{1,1}', c_{2,1}', \dots, c_{h-1,1}'$, 
    \item for each $i, j \in [h] \setminus \{1,h\}$, delete the edges $r_{i,h-1}'r_{i, h}$ and $c_{h-1,j}'r_{h,j}$,
    \item for each $i, j \in [h] \setminus \{1\}$, add the edge $r'_{i,j-1}c_{i-1,j}'$, 
    \item for each $i, j \in [h] \setminus\{1,h\}$, replace $r_{i,j}$ by a constant-size planar graph which will be defined in the proof of Theorem~\ref{thm:lowerbound_new}. 
\end{itemize}
See Figure~\ref{fig:grid} for an illustration.    

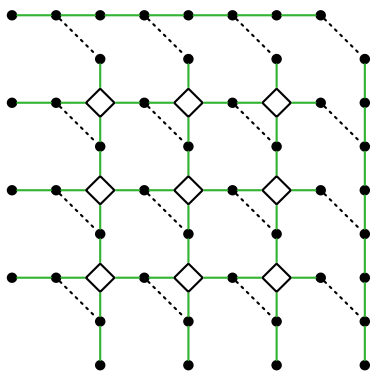
\begin{figure}[ht!]\centering
\setlength{\belowcaptionskip}{-10pt}
  \begin{tikzpicture}[scale=.58] \tikzset{arrowb/.style={darkgreen},
        arrowr/.style={dotted},
        state/.style = {draw,fill,circle,inner sep = 1.1pt},
        state1/.style = {draw,diamond, inner sep = 2.75pt}}
        
    \clip(-.22,-.22)rectangle(8.22,8.22);
    \foreach \x in {1,...,3}
    \foreach \y in {1,...,3} 
       {\pgfmathtruncatemacro{\label}{\x + 5 *  \y }
       \node [state1]  (2\x\y) at (2*\x,2*\y) {};}
       
    \foreach \x in {0}
    \foreach \y in {1,...,4} 
       {\pgfmathtruncatemacro{\label}{\x + 5 *  \y }
       \node [state]  (2\x\y) at (2*\x,2*\y) {};}
    
    \foreach \x in {4}
    \foreach \y in {0,...,3} 
       {\pgfmathtruncatemacro{\label}{\x + 5 *  \y }
       \node [state]  (2\x\y) at (2*\x,2*\y) {};}
       
    \foreach \x in {1,...,3}
    \foreach \y in {0} 
       {\pgfmathtruncatemacro{\label}{\x + 5 *  \y }
       \node [state]  (2\x\y) at (2*\x,2*\y) {};}
    
    \foreach \x in {1,...,3}
    \foreach \y in {4} 
       {\pgfmathtruncatemacro{\label}{\x + 5 *  \y }
       \node [state]  (2\x\y) at (2*\x,2*\y) {};}
    
    \foreach \x in {1,...,4}
    \foreach \y in {0,...,3} 
       {\pgfmathtruncatemacro{\label}{\x + 5 *  \y }
       \node [state]  (\x\y') at (2*\x,2*\y+1) {};}
    \foreach \x in {0,...,3}
    \foreach \y in {1,...,4} 
      {\pgfmathtruncatemacro{\label}{\x + 5 *  \y }
      \node [state]  (\x'\y) at (2*\x+1,2*\y) {};}

    \foreach \x [count=\xi] in {0,...,3}
    {
    \foreach \y [count=\yi] in {0,...,3}  
      {
        \path 
         (\x'\yi) edge [arrowr] (\xi\y');
      }
    }

    \foreach \y in {1,...,3} {
        \foreach \x [count=\xi] in {0,...,2} {
            \draw[arrowb] (2\x\y) -- (\x'\y)-- (2\xi\y);
        }
        \draw[arrowb] (23\y) -- (3'\y);
    }
    
    \foreach \x  in {1,...,3} {
        \foreach \y [count=\yi] in {0,...,2} {
            \draw [arrowb] (2\x\y) -- (\x\y')-- (2\x\yi);
        }
        \draw [arrowb](2\x3)  -- (\x3');
    }
    
    \foreach \x  in {4}
    {
    \foreach \y [count=\yi] in {0,...,2}  
      {
        \path 
        (2\x\y) edge [arrowb] (\x\y')
        (\x\y') edge [arrowb] (2\x\yi);
      }
    }
    \foreach \x [count=\xi] in {0,...,2}
    {
    \foreach \y in {4}  
      {
        \path 
        (2\x\y) edge [arrowb] (\x'\y)
        (\x'\y) edge [arrowb] (2\xi\y);
      }
    }
    \path 
      (234) edge [arrowb] (3'4)
      (243) edge [arrowb] (43')
    ;
 \end{tikzpicture}
\caption{The grid with gadgets of height $4$. The diamond vertices are the vertices to be replaced.}\label{fig:grid}
\end{figure}

We are now ready to prove our lower bound for minor-closed classes.

\begin{theorem} \label{thm:lowerbound_new}
For every proper minor-closed class $\mathcal{C}$, there exists a constant $c'>0$ such that for every $n$-vertex graph $G \in \mathcal{C}$, %
\begin{equation*}
\xc(\COR(G)) \geq 2^{c' \tw(G)}.
\end{equation*}
\end{theorem}

\begin{proof}
As shown by Demaine and Hajiaghayi~\cite{DH08} (see also Kawarabayashi and Kobayashi~\cite{KK12}), since our initial graph $G$ belongs to a proper minor-closed class, $G$ contains $G_{t,t}$ as a minor, where $t = \Omega(\tw(G))$. 
Moreover, observe that there exists $h=\Omega(t)$ such that if $H$ is the grid with gadgets of height $h$, then  $H \preceq G_{t,t}$. Let $H$ be the grid with gadgets of height $h$.  Since the extension complexity of the correlation polytope is minor monotone, $\xc(\COR(H)) \le \xc(\COR(G))$. Therefore, it suffices to show the theorem for $H$.

The idea is to describe a face of $\COR(H)$ which projects to the correlation polytope of the complete bipartite graph $K_{h,h}$. 
The projection is defined in such a way that the vertices of the bipartition of $K_{h,h}$ correspond to bottom and left vertices of the grid with gadgets, and the edges of $K_{h,h}$ correspond to the dotted diagonal edges of Figure~\ref{fig:grid}.
Roughly speaking, the face is defined in such a way that the value of the variable for each bottom and each left vertex of the grid propagates along the corresponding vertical and horizontal path. The gadgets make sure that propagation along vertical paths does not interfere with propagation along horizontal paths. 

The gadget used in the proof is inspired by the \emph{crossover gadget} of the reduction from \textsf{3-SAT} to \textsf{Planar 3-SAT} in~\cite{lichtenstein1982planar}.
In the crossover gadget (see Figure~\ref{fig:sub2_gadget}) the square vertices represent the clauses of a \textsf{SAT} formula and the round vertices represent the variables. When a round vertex is adjacent to a square via a blue-dashed (resp.~red-dotted) edge, this means that the corresponding variable (resp.~negation of the variable) appears in the corresponding clause. The main property of the crossover gadget in Figure~\ref{fig:sub2_gadget} is that if all clauses in the crossover gadget are satisfied, then $b$ is true if and only if $t$ is true and $\ell$ is true if and only if $r$ is true. 
 
We simulate each copy of the crossover gadget with a copy of a fixed planar graph, which we call the \emph{graph gadget}. Each copy is independent from the other copies: all the new vertices and edges that are created in the process are local to the copy. We ensure that the graph gadget behaves exactly like the crossover gadget by restricting to an appropriate face of $\COR(H)$.  For each variable $i$ in the crossover gadget, we interpret $i$ as being true or false, according as, the corresponding variable $x_i$ of $\COR(H)$  is equal to $1$ or $0$.
The grid with gadgets $H$ in Figure~\ref{fig:grid} is obtained by replacing each diamond by the graph gadget.  

\begin{figure}[ht!]\centering 
  \begin{subfigure}{.42\textwidth}\centering
    \begin{tikzpicture}[scale=.68]
\tikzset{arrowg/.style={darkgreen},
    state/.style = {draw,diamond,inner sep = 2.75pt}}
    \clip(-3.8,-3.8)rectangle(3.8,3.8);
    \node[state] (0) at (0,0) {};
    \node (r) at (2,0) {};
    \node (l) at (-2,0) {};
    \node (d) at (0,-2) {};
    \node (u) at (0,2) {};
    
    \path 
        (0) edge[arrowg] (u)
        (0) edge[arrowg] (r)
        (0) edge[arrowg] (l)
        (0) edge[arrowg] (d);
  \end{tikzpicture}
  \caption{}\label{fig:sub1_gadget}
  \end{subfigure}\hspace*{1cm}
  \begin{subfigure}{.42\textwidth}\centering
    \begin{tikzpicture}[scale=.65]   
\tikzset{arrowr/.style={red,dotted},
    arrowb/.style={blue,densely dashed},
    arrowg/.style={darkgreen},
    state/.style = {draw,rectangle,inner sep = 2.2pt},
    state1/.style = {draw,fill,circle,inner sep = 1.2pt}}
    \clip(-3.8,-3.8)rectangle(3.8,3.8);
    \node[state1, label={[left,yshift=-1mm]$b$}] (b2) at (0,-3) {};
    \node[state1, label={[left,yshift=1mm]$t$}] (b1) at (0,3) {};
    \node[state1, label={[above,xshift=-1mm]$\ell$}] (a1) at (-3,0) {};
    \node[state1, label={[above,xshift=1mm]$r$}] (a2) at (3,0) {};
    \node (r) at (4,0) {};
    \node (l) at (-4,0) {};
    \node (d) at (0,-4) {};
    \node (u) at (0,4) {};
    
    \node[state1] (alpha) at (1,-1) {};
    \node[state1] (beta) at (1,1) {};
    \node[state1] (delta) at (-1,-1) {};
    \node[state1] (gamma) at (-1,1) {};
    
    \node[state] (-22) at (-1.5,1.5) {};
    \node[state] (-2-2) at (-1.5,-1.5) {};
    \node[state] (22) at (1.5,1.5) {};
    \node[state] (2-2) at (1.5,-1.5) {};
    \node[state] (-21) at (-2,.5) {};
    \node[state] (-2-1) at (-2,-.5) {};
    \node[state] (21) at (2,.5) {};
    \node[state] (2-1) at (2,-.5) {};
    \node[state] (1-2) at (.5,-2) {};
    \node[state] (-1-2) at (-.5,-2) {};
    \node[state] (12) at (.5,2) {};
    \node[state] (-12) at (-.5,2) {};
    \node[state1] (0) at (0,0) {};
    \node[state] (0+) at (0,.6) {};
    \node[state] (0-) at (0,-.6) {};
    \node[state] (0-1) at (0,-1) {};
    \node[state] (01) at (0,1) {};
    \node[state] (-10) at (-1,0) {};
    \node[state] (10) at (1,0) {};
    
    \path 
         (b1) edge [arrowr] (-12)
         (b1) edge [arrowr] (12)
         (b1) edge [arrowb] (-22)
         (b1) edge [arrowb] (22)
         (b2) edge [arrowr] (2-2)
         (b2) edge [arrowr] (-2-2)
         (b2) edge [arrowb] (1-2)
         (b2) edge [arrowb] (-1-2)
         (a1) edge [arrowr] (-21)
         (a1) edge [arrowr] (-2-1)
         (a1) edge [arrowb] (-2-2)
         (a1) edge [arrowb] (-22)
         (a2) edge [arrowr] (2-2)
         (a2) edge [arrowr] (22)
         (a2) edge [arrowb] (2-1)
         (a2) edge [arrowb] (21)
         (alpha) edge [arrowr] (10)
         (alpha) edge [arrowr] (2-1)
         (alpha) edge [arrowb] (2-2)
         (alpha) edge [arrowr] (1-2)
         (alpha) edge [arrowr] (0-1)
         (alpha) edge [arrowb] (0-)
         (beta) edge [arrowb] (0+)
         (beta) edge [arrowr] (01)
         (beta) edge [arrowr] (12)
         (beta) edge [arrowb] (22)
         (beta) edge [arrowr] (21)
         (beta) edge [arrowr] (10)
         (delta) edge [arrowb] (0-)
         (delta) edge [arrowr] (0-1)
         (delta) edge [arrowr] (-1-2)
         (delta) edge [arrowb] (-2-2)
         (delta) edge [arrowr] (-2-1)
         (delta) edge [arrowr] (-10)
         (gamma) edge [arrowb] (0+)
         (gamma) edge [arrowr] (-10)
         (gamma) edge [arrowr] (-21)
         (gamma) edge [arrowb] (-22)
         (gamma) edge [arrowr] (-12)
         (gamma) edge [arrowr] (01)
         (a1) edge [arrowg] (l)
         (a2) edge [arrowg] (r)        
         (b1) edge [arrowg] (u)
         (b2) edge [arrowg] (d)
         (0) edge [arrowb] (0-)
         (0) edge [arrowr] (0+);
\end{tikzpicture}
  \caption{}\label{fig:sub2_gadget}
  \end{subfigure}
  \setlength{\belowcaptionskip}{-10pt}
\caption{Crossover gadget.}\label{fig:gadget}
\end{figure}
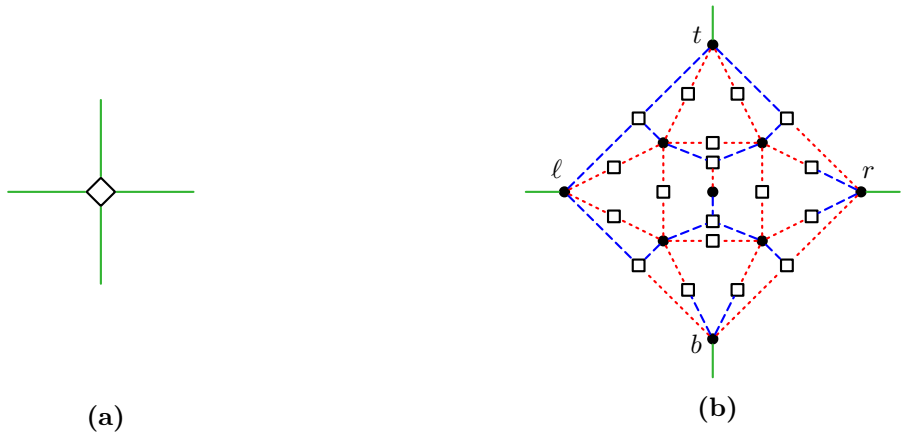

To convert the crossover gadget into the graph gadget, we describe replacement rules for the neighbourhood of each square vertex (clause).   
After all the rules are applied, we obtain just a graph. The rules are described in Figures~\ref{fig:not2and} to~\ref{fig:3or}, together with a set of equations that are to be added to the description of the desired face of $\COR(H)$. The equations we add originate from valid inequalities, hence the set of points of $\COR(H)$ satisfying all of them is indeed a face of $\COR(H)$.

First, consider Figure~\ref{fig:not2and}. It specifies that degree-$2$ square vertices adjacent to two red edges should be replaced by a single edge. Here, we want to simulate the clause $(\neg\,i \vee \neg\,j) = \neg\,(i \wedge j)$. Thus, we want at most one of $i$ or $j$ to be true, which is equivalent to setting $x_{ij} = 0$. 

Next, the replacement described in Figure~\ref{fig:noti} is to transform paths consisting of one red edge and one blue edge incident to a degree-$2$ square vertex into a path of the same length. Here, we want to simulate $(\neg \,i \vee  j)$. We claim this is equivalent to adding the following set of equations: $x_{i\overline{\jmath}} = 0$ and $x_j+x_{\overline{\jmath}}-2x_{j\overline{\jmath}} = 1$. The first constraint implies that $x_i$ and $x_{\overline{\jmath}}$ do not both equal $1$,
while the second constraint ensures that exactly one of $x_j$ and $x_{\overline{\jmath}}$ equals $1$.
Therefore, at least one of $\neg \, i$ or $j$ is true.  

Finally, we replace the degree-$3$ square vertices as in Figures~\ref{fig:3or_ij}, \ref{fig:3or_i} and \ref{fig:3or}. 
Here, we want to simulate the clauses containing three literals.  For instance, consider Figure~\ref{fig:3or_ij}, where we want to simulate $(\neg\,i \vee \neg\,j \vee k)$. The replacements in Figure~\ref{fig:3or_i} and \ref{fig:3or} are similar. We claim this is described by adding the following set of equations:
$x_{ii'} = 0$, $x_{jj'} = 0$, $x_{\overline{k}k'} = 0$, $x_k+x_{\overline{k}}-2x_{k\overline{k}} = 1$, and
$x_{i'} + x_{j'} + x_{k'} - 2x_{i'j'} - 2x_{i'j'} - 2x_{j'k'} = 1$. The last constraint implies that exactly one of $x_{i'}, x_{j'}$,  or $x_{k'}$ equals $1$.  The second constraint implies that exactly one of $x_k$ and $x_{\overline{k}}$ equals $1$.  The first constraint implies that at most one of $x_i$ and $x_{i'}$ equals $1$, at most one of $x_j$ and $x_{j'}$ equals $1$, and at most one of $x_k$ and $x_{k'}$ equals $1$.  Together, this implies $x_i=0$ or $x_j=0$ or $x_k=1$, as required.

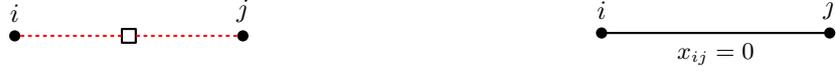
\begin{figure}[ht!]\centering
  \begin{subfigure}{.4\textwidth}\centering
    \begin{tikzpicture}[scale=.75]     \tikzset{arrow/.style={red,dotted},
    state/.style = {draw,circle,inner sep = 1.2pt, fill},
    state1/.style = {draw,rectangle,inner sep = 2.6pt}}
    \clip(-2.5,-.6) rectangle (2.5,.7);
    \node[state1] (0) at (0,0) {};
    \node[inner sep = 1pt, state, label=above:{$i$}] (i) at (-2,0) {};
    \node[inner sep = 1pt, state, label=above:{$j$}] (j) at (2,0) {};
    \path 
         (0) edge [arrow] (i)
         (0) edge [arrow] (j);
  \end{tikzpicture}
  \end{subfigure}\hspace*{1cm}
  \begin{subfigure}{.4\textwidth}\centering
    \begin{tikzpicture}[scale=.75]
\tikzset{state/.style = {draw,fill,circle,inner sep = 1.2pt}}
    \clip(-2.5,-.6) rectangle (2.5,.6);

    \node[state, label=above:{$i$}] (i) at (-2,0) {};
    \node[state, label=above:{$j$}] (j) at (2,0) {};
    \node [below] at (0,0) {\footnotesize $x_{ij} = 0$};
    
    \draw (i) -- (j);
  \end{tikzpicture}
  \end{subfigure}\vspace*{-1mm}
\caption{Replacement for $(\neg\,i \vee \neg\,j)$.}\label{fig:not2and}
\end{figure}\vspace*{-5mm}

\begin{figure}[ht!]\centering 
  \begin{subfigure}{.4\textwidth}\centering
    \begin{tikzpicture}[scale=.75] \tikzset{arrowr/.style={red,dotted},
    arrowb/.style={blue, densely dashed},
    state/.style = {draw,circle,inner sep = 1.2pt, fill},
    state1/.style = {draw,rectangle,inner sep = 2.6pt}}
    \clip(-2.5,-1.05)rectangle(2.5,.8);
    \node[state1] (0) at (0,0) {};
    \node[inner sep = 1pt, state, label=above:{$i$}] (i) at (-2,0) {};
    \node[inner sep = 1pt, state, label=above:{$j$}] (j) at (2,0) {};
    \path 
         (0) edge [arrowr] (i)
         (0) edge [arrowb] (j);
\end{tikzpicture}
  \end{subfigure}\hspace*{1cm}
  \begin{subfigure}{.4\textwidth}\centering
    \begin{tikzpicture}[scale=.75]
\tikzset{arrowb/.style={blue, densely dashed},
    state/.style = {draw,fill,circle, inner sep = 1.2pt}}
    \clip(-2.5,-1.05)rectangle(2.5,.8);
    
    \node[state] (0) at (0,0) {};
    \node[state, label=above:{$\overline{\jmath}$}] (notj) at (0,0) {};  
    
    \node[below,text width=4.5cm,align=center] at (0.25,0) {\footnotesize 
    $\begin{array}{r@{\hskip2pt}l}
    x_{i\overline{\jmath}} = 0 \\[-2pt]
    x_j+x_{\overline{\jmath}}-2x_{j\overline{\jmath}} = 1
    \end{array}$};
    
    \node[state, label=above:{$i$}] (i) at (-2,0) {};
    \node[state, label=above:{$j$}] (j) at (2,0) {};
    
    \draw (i) -- (notj) -- (j);
\end{tikzpicture}
  \end{subfigure}\vspace*{-1mm}
\caption{Replacement for $(\neg\,i \vee j)$.}
\label{fig:noti}
\end{figure}
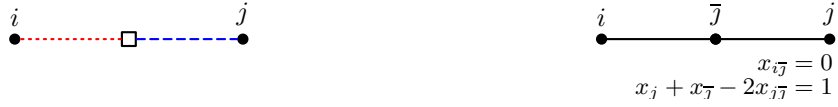\vspace*{-3mm}

\begin{figure}[ht!]\centering
    \savebox{\largestimage}{
        \begin{tikzpicture}[scale=1.15] \tikzset{state/.style = {draw,fill,circle, inner sep = 1.2pt}}
                
    \clip(-3.75,-1.5)rectangle(1.8,2.2);
    \node[state, label={[yshift=1.5mm]left:$i'$}] (i') at (0,0) {};
    \node[state, label={[yshift=1.5mm]right:$j'$}] (j') at (1,0) {};
    \node[state,label=right:{$k'$}] (k') at (.5,.867) {};
    \node[state, label={[yshift=1.5mm]left:$i$}] (i) at (-.5,-.289) {};
    \node[state, label={[yshift=1.5mm]right:$j$}] (j) at (1.5,-.289) {};
    \node[state,label=right:{$\overline{k}$}] (notk) at (.5,1.44) {};
    \node[state, label=right:{$k$}] (k) at (.5,2.02) {};
    
    \node[below,text width=7.5cm,align=center] at (-1,-.289) {\footnotesize
    $\begin{array}{r@{\hskip2pt}l}
    x_{ii'} = x_{jj'} = x_{\overline{k}k'} &= 0 \\
    x_k + x_{\overline{k}} - 2x_{k\overline{k}} &= 1\\
    x_{i'} + x_{j'} + x_{k'} - 2x_{i'j'} - 2x_{i'k'} - 2x_{j'k'} &= 1
    \end{array}$
    };
    
    \draw (i)--(i')--(k');
    \draw (j)--(j');
    \draw (k)--(notk)--(k');
    \draw (i')--(j')--(k');
\end{tikzpicture}
    }
    \begin{subfigure}[t]{.45\textwidth}\centering
        \raisebox{\dimexpr.5\ht\largestimage-.5\height}{ \begin{tikzpicture}[scale=2] \tikzset{arrowr/.style={red,dotted},     arrowb/.style={blue, densely dashed},
    state/.style = {draw,fill,circle,inner sep = 1.2pt},
    state1/.style = {draw,rectangle, inner sep = 2.6pt}}

    \node[state1] (0) at (.5,.289) {};
    \node[state, label=left:{$i$}] (i) at (0,0) {};
    \node[state, label=right:{$j$}] (j) at (1,0) {};
    \node[state, label=right:{$k$}] (k) at (.5,.867) {};

    \path 
         (0) edge [arrowr] (i)
         (0) edge [arrowb] (k)
         (0) edge [arrowr] (j);
\end{tikzpicture}
        }
    \end{subfigure}
    \begin{subfigure}[t]{.45\textwidth}\centering
        \hspace*{-1.5cm}\usebox{\largestimage}
    \end{subfigure}
\caption{Replacement for $(\neg\,i \vee \neg\,j \vee k)$}\label{fig:3or_ij}
\end{figure}
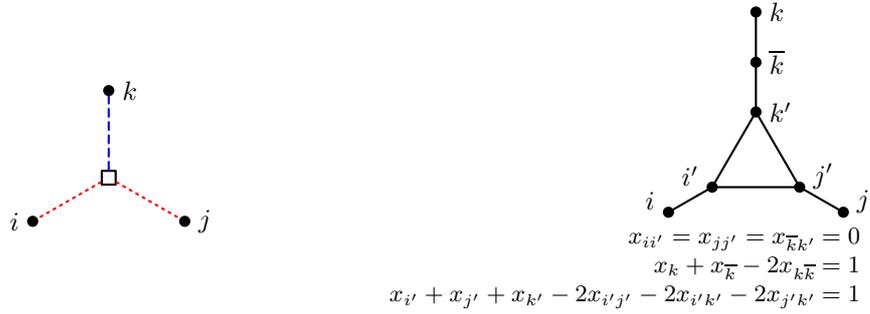

\begin{figure}[ht!]\centering
    \savebox{\largestimage}{
        \begin{tikzpicture}[scale=1.15]
\tikzset{state/.style = {draw,fill,circle, inner sep = 1.2pt}}
    \clip(-3.35,-2.05)rectangle(2.3,2.2);
    \node[state, label={[yshift=1.5mm]left:$i'$}] (i') at (0,0) {};
    \node[state, label={[yshift=1.5mm]right:$j'$}] (j') at (1,0) {};
    \node[state,label=right:{$k'$}] (k') at (.5,.867) {};
    \node[state, label={[yshift=1.5mm]left:$i$}] (i) at (-.5,-.289) {};
    \node[state, label={[yshift=1.5mm]right:$\overline{\jmath}$}] (notj) at (1.5,-.289) {};
    \node[state,label=right:{$\overline{k}$}] (notk) at (.5,1.44) {};
    \node[state, label=right:{$k$}] (k) at (.5,2.02) {};
    \node[state, label={[yshift=1.5mm]right:$j$}] (j) at (2,-.578) {};
    
    \node[below,text width=7.5cm,align=center] at (-.5,-.578) {\footnotesize
    $\begin{array}{r@{\hskip2pt}l}
    x_{ii'}=  
    x_{\overline{\jmath}j'} = 
    x_{\overline{k}k'} &= 0 \\
    x_j + x_{\overline{\jmath}} - 2x_{j\overline{\jmath}} &= 1\\
    x_k + x_{\overline{k}} - 2x_{k\overline{k}} &= 1\\
    x_{i'} + x_{j'} + x_{k'} - 2x_{i'j'} - 2x_{i'k'} - 2x_{j'k'} &= 1
    \end{array}$
    };
    
    \draw (i)--(i')--(k');
    \draw (j)--(notj)--(j');
    \draw (k)--(notk)--(k');
    \draw (i')--(j')--(k');
\end{tikzpicture}
    }
    \begin{subfigure}[t]{.45\textwidth}\centering
        \raisebox{\dimexpr.5\ht\largestimage-.5\height}{
            \begin{tikzpicture}[scale=2] 
\tikzset{arrowr/.style={red,dotted},     arrowb/.style={blue, densely dashed},
    state/.style = {draw,fill,circle,inner sep = 1.2pt},
    state1/.style = {draw,rectangle, inner sep = 2.6pt}}

    \node[state1] (0) at (.5,.289) {};
    \node[state, label=left:{$i$}] (i) at (0,0) {};
    \node[state, label=right:{$j$}] (j) at (1,0) {};
    \node[state, label=right:{$k$}] (k) at (.5,.867) {};

    \path 
         (0) edge [arrowr] (i)
         (0) edge [arrowb] (k)
         (0) edge [arrowb] (j);
\end{tikzpicture}
        }
    \end{subfigure}
    \begin{subfigure}[t]{.45\textwidth}\centering
        \hspace*{-1.5cm}\usebox{\largestimage}
    \end{subfigure}
\caption{Replacement for $(\neg\,i \vee j \vee k)$}\label{fig:3or_i}
\end{figure}
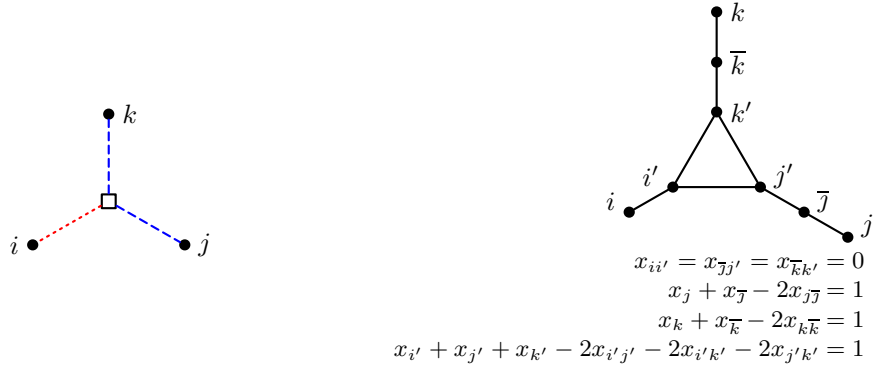
{\setlength{\belowcaptionskip}{-10pt}
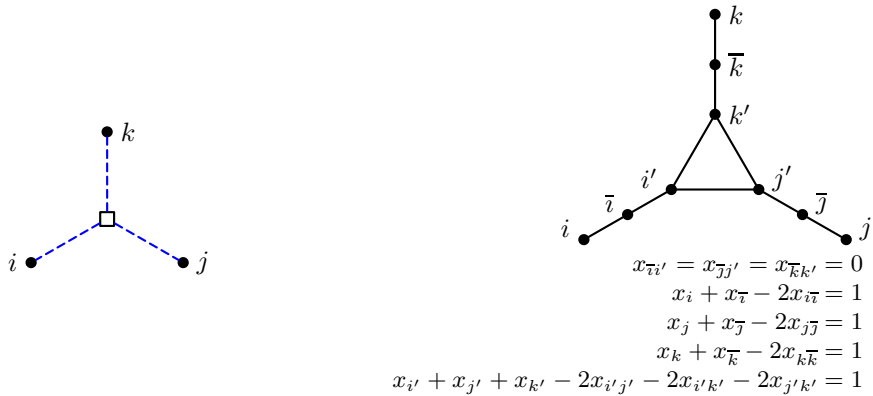
\begin{figure}[ht!]\centering
        \savebox{\largestimage}{
        \begin{tikzpicture}[scale=1.15]
\tikzset{state/.style = {draw,fill,circle, inner sep = 1.2pt}}
    \clip(-3.35,-2.4)rectangle(2.3,2.2);
    \node[state, label={[yshift=1.5mm]left:$i'$}] (i') at (0,0) {};
    \node[state, label={[yshift=1.5mm]right:$j'$}] (j') at (1,0) {};
    \node[state,label=right:{$k'$}] (k') at (.5,.867) {};
    \node[state, label={[yshift=1.5mm]left:$\overline{\imath}$}] (noti) at (-.5,-.289) {};
    \node[state, label={[yshift=1.5mm]right:$\overline{\jmath}$}] (notj) at (1.5,-.289) {};
    \node[state,label=right:{$\overline{k}$}] (notk) at (.5,1.44) {};
    \node[state, label=right:{$k$}] (k) at (.5,2.02) {};
    \node[state, label={[yshift=1.5mm]right:$j$}] (j) at (2,-.578) {};
    \node[state, label={[yshift=1.5mm]left:$i$}] (i) at (-1,-.578) {};
    
    \node[below,text width=7.5cm,align=center] at (-.5,-.578) {\footnotesize
    $\begin{array}{r@{\hskip2pt}l}
    x_{\overline{\imath}i'} = 
    x_{\overline{\jmath}j'} =
    x_{\overline{k}k'} &= 0 \\
    x_i + x_{\overline{\imath}} - 2x_{i\overline{\imath}} &= 1\\
    x_j + x_{\overline{\jmath}} - 2x_{j\overline{\jmath}} &= 1\\
    x_k + x_{\overline{k}} - 2x_{k\overline{k}} &= 1\\
    x_{i'} + x_{j'} + x_{k'} - 2x_{i'j'} - 2x_{i'k'} - 2x_{j'k'} &= 1
    \end{array}$
    };
    
    \draw (i)--(noti)--(i')--(k');
    \draw (j)--(notj)--(j');
    \draw (k)--(notk)--(k');
    \draw (i')--(j')--(k');
\end{tikzpicture}
    }
    \begin{subfigure}[t]{.45\textwidth}\centering
        \raisebox{\dimexpr.5\ht\largestimage-.5\height}{
            \begin{tikzpicture}[scale=2] \tikzset{arrowr/.style={red,dotted},     arrowb/.style={blue, densely dashed},
    state/.style = {draw,fill,circle,inner sep = 1.2pt},
    state1/.style = {draw,rectangle, inner sep = 2.6pt}}

    \node[state1] (0) at (.5,.289) {};
    \node[state, label=left:{$i$}] (i) at (0,0) {};
    \node[state, label=right:{$j$}] (j) at (1,0) {};
    \node[state, label=right:{$k$}] (k) at (.5,.867) {};

    \path 
         (0) edge [arrowb] (i)
         (0) edge [arrowb] (k)
         (0) edge [arrowb] (j);
\end{tikzpicture}
        }
    \end{subfigure}
    \begin{subfigure}[t]{.45\textwidth}\centering
        \hspace*{-1.5cm}\usebox{\largestimage}
    \end{subfigure}
\caption{Replacement for $(i \vee j \vee k)$}\label{fig:3or}
\end{figure}}

After these replacements we obtain a constant size graph gadget. The above constraints define a face of $\COR(H)$ such that all clauses inside the crossover gadgets are satisfied.
Therefore, by the main property of the crossover gadget (see~\cite{lichtenstein1982planar} for a proof), $x_b = x_t$ and $x_\ell = x_r$ (see Figure~\ref{fig:gadget}).  

If $ij$ is a vertical or horizontal edge of $H$ (see the solid green edges in Figure~\ref{fig:grid}), then we set $x_i = x_j = x_{ij}$ (note that $x_i \leq x_{ij}$ and $x_j \leq x_{ij}$ are both valid inequalities). Together with the crossover gadgets, this ensures that if $i$ is a vertex on the bottom of $H$, then the value of $x_i$ propagates vertically. Similarly, if $j$ is a vertex on the left of $H$ then the value of $x_j$ propagates horizontally.  It follows that by projecting onto the diagonal dotted edges in Figure~\ref{fig:grid}, 
the face of $\COR(H)$ described above projects to $\COR(K_{h,h})$. 
Finally, since $K_{h}$ is a minor of $K_{h,h}$, we have $\xc(\COR(K_{h})) \le \xc(\COR(K_{h,h}))$.
In~\cite{KW15}, it is shown that $\xc(\COR (K_{h})) \geqslant (1.5)^h$. (A weaker exponential bound was given earlier in~\cite{fiorini2011exponential}.) Since $h = \Omega(\tw(G))$, this completes the proof.
\end{proof}

\begin{proof}[Proof of Theorem~\ref{lowerbound}]
Let $c'$ be the constant from Theorem~\ref{thm:lowerbound_new} and 
set $c$ to be equal to $\min\{c'/2,1/2\}$. By taking the geometric mean of the bounds in Observation~\ref{dimension} and Theorem~\ref{thm:lowerbound_new} we obtain Theorem~\ref{lowerbound}.
\end{proof}

\subsection*{Acknowledgements} We thank Mark Rowland and Adrian Weller for interesting discussions on the connection between the correlation polytope and graph minors.  We also thank the referee for their careful reading of the paper.  This project is supported by ERC grant \emph{FOREFRONT} (grant agreement no. 615640) funded by the European Research Council under the EU's 7th Framework Programme (FP7/2007-2013).

\bibliographystyle{amsplain}
\bibliography{main.bib}

\end{document}